\documentclass[journal]{IEEEtran}

\usepackage{microtype}
\usepackage{graphicx}
\usepackage[caption=false,font=footnotesize]{subfig} % Standard IEEE way to handle subfigures
\usepackage{booktabs} % for professional tables
\usepackage{hyperref}
\usepackage{cite} % Recommended for IEEE bibliography management
\usepackage{geometry}
\usepackage{amsmath}
\usepackage{amssymb}
\usepackage{mathtools}
\usepackage{amsthm}

\usepackage[ruled,vlined,linesnumbered]{algorithm2e}

\usepackage[capitalize,noabbrev]{cleveref}
\makeatletter
\let\ps@IEEEtitlepagestyle\ps@empty
\let\ps@plain\ps@empty
\makeatother

\theoremstyle{plain}
\newtheorem{theorem}{Theorem}[section]

\theoremstyle{definition}

\theoremstyle{remark}

\begin{document}
\pagestyle{empty}
\thispagestyle{empty}
\title{Mitigating Anchoring Bias in LLM-Based Agents for\
Energy-Efficient 6G Autonomous Networks}

\author{
\IEEEauthorblockN{Hatim Chergui\IEEEauthorrefmark{1},
Claudia Carballo Gonz\'{a}lez\IEEEauthorrefmark{1},
Farhad Rezazadeh\IEEEauthorrefmark{2},
Merouane Debbah\IEEEauthorrefmark{3}}

\IEEEauthorblockA{\IEEEauthorrefmark{1}i2CAT Foundation, 08034 Barcelona, Spain}

\IEEEauthorblockA{\IEEEauthorrefmark{2}Universitat Polit\`ecnica de Catalunya (UPC), 08034 Barcelona, Spain}

\IEEEauthorblockA{\IEEEauthorrefmark{3}Research Institute for Digital Future, Khalifa University, 127788 Abu Dhabi, UAE}
}

\maketitle

\begin{abstract}
This paper presents an autonomous agentic resource negotiation framework designed to enable zero-touch network slicing in 6G architectures using Large Language Model (LLM) agents. While LLMs offer powerful reasoning capabilities, we demonstrate that such agents inherently suffer from anchoring bias, rigidly adhering to initial heuristic proposals and causing severe network over-provisioning. To systematically mitigate this cognitive bias, we propose a novel randomized anchoring strategy modeled via a Truncated 3-Parameter Weibull distribution. This mathematically bounded approach seamlessly integrates with burst-aware Digital Twins (DTs) employing Conditional Value at Risk (CVaR) to rigorously guarantee strict Service Level Agreement (SLA) tail-latencies. To validate our methodology, we introduce and prove the \emph{Bimodal Constraint-Avoidance Utility Theorem}, demonstrating that while feasible negotiations follow classical convex bounds, highly constrained scenarios undergo a phase transition governed by an inverse rational decay envelope. Empirical results generated using a locally hosted 1B-parameter model (\texttt{otel-llm-1b-it}) confirm these dual-regime bounds. Our cognitive de-biasing successfully dismantles rigid negotiation patterns, forcing agents into active exploration to safely ride SLA boundaries and boost system energy savings up to 25\%. Crucially, the lightweight 1B LLM achieves sub-second inference latencies (0.95s mean), ensuring our multi-agent framework is compatible with the operational timescales of the O-RAN non-Real-Time RAN Intelligent Controller (non-RT RIC)\footnote{Our source code is available for non-commercial use at \url{https://github.com/HatimChergui}.}.
\end{abstract}


\section{Introduction}

The transition toward sixth-generation (6G) networks is significantly increasing the complexity of managing ubiquitous connectivity and stringent service requirements, exposing the limitations of current network automation approaches. In this context, 6G wireless systems are being driven toward a vision of operational self-governance. To meet the rigorous demands of TM Forum's Level 4 and Level 5 autonomy---representing closed-loop and full autonomy, respectively \cite{tmforum2021autonomous}---network architectures move beyond conventional automation paradigms. This evolution necessitates the deployment of \emph{agentic systems} \cite{ferrag2025llmreasoningautonomousai}. Unlike legacy controllers, these Large Language Model (LLM)-driven entities are designed to reason, plan, and negotiate at a high-level objective space, enabling the dynamic management of slice orchestration and service assurance in highly dynamic environments. However, delegating critical network operations to autonomous agents introduces new challenges related to reliability, robustness, and decision integrity.

Recent scholarship suggests a troubling phenomenon in these advanced architectures: artificial intelligence (AI) agents frequently exhibit cognitive biases that mirror human psychological distortions. Rooted in the foundational work of Tversky and Kahneman on heuristics and systematic errors \cite{kahneman}, these biases can compromise collective decision-making, fairness, and safety of agent-augmented 6G systems. As Xie et al. \cite{Xie2024} observe, these distortions are particularly prevalent in multi-agent systems, which serve as the blueprint for decentralized 6G architectures. Their impact propagates across the entire functional pipeline, manifesting in four key layers. At the data level, biases may arise from historical or cultural imbalances in training datasets, for instance leading to \emph{legacy bias} where agents fail to fully exploit advanced 6G capabilities. At the prompt level, framing effects can skew decision-making, such as prioritizing spectral efficiency at the expense of energy efficiency or fairness. During reasoning, agents may rely on flawed heuristics \cite{rezazadeh2025agentic}, exhibiting behaviours such as availability-driven over-provisioning or confirmation bias in threat detection. Finally, biases also affect tool and memory integration, where recency, primacy, or authority effects can distort the use of historical data and external information sources.

Consequently, the intersection of cognitive psychology and autonomous networking has become a critical research direction. A foundational contribution is provided by Chergui et al. \cite{chergui2025tutorialcognitivebiasesagentic}, who presented a structured tutorial on cognitive biases within 6G agentic systems. Their work establishes the mathematical formulations for these biases and proposes mitigation strategies at both the agent and system levels, supported by practical 6G use cases. 

Beyond individual agent errors, research has highlighted the risks inherent in multi-agent interaction. In \cite{unmasking}, the authors demonstrate that iterative agent discussions can amplify existing biases, creating \emph{conversational echo chambers} where agents prematurely converge on a skewed consensus. Furthermore, in~\cite{wang2025biasguarrd}, the authors study cognitive biases in LLM-based interpersonal conflict resolution, showing that model judgments shift under biased prompt phrasing. They propose BiasGUARRD, a multi-agent framework that detects and mitigates such biases in socially sensitive decision-making. In \cite{oh2025understanding}, the authors likewise analyze limitations of Multi-Agent Debate for LLM reasoning, showing that it can reinforce biases, and propose a refined multi-agent prompting framework that enhances reasoning diversity and reduces bias, leading to improved decision accuracy and robustness across strategic tasks. Empirical validation of these concerns is provided by Knipper et al. \cite{knipper2025biasdetailsassessmentcognitive}, who indicate that while larger models ($>$32B parameters) tend to reduce bias in approximately 39.5\% of cases, more detailed prompting---while generally helpful---can actually increase certain errors, such as Overattribution, by up to 8.8\%.

In light of these challenges, this paper makes the following key contributions:
\begin{itemize}
    \item \textbf{Anchoring Bias \& Tail-Risk Profiling:} We demonstrate how initial proposals trap agents in rigid, over-provisioned optima. To ensure strict SLA compliance under stochastic traffic bursts, we integrate Conditional Value at Risk (CVaR) into the agent's Digital Twin (DT).
    \item \textbf{Bimodal Constraint-Avoidance Theorem:} We introduce and rigorously prove a novel theoretical framework describing a phase transition in utility degradation bounds. We map empirical utility loss to show dual-regime behavior: a classical linear bound for feasible conditions and an inverse rational decay bound for constraint-heavy environments.
    \item \textbf{Truncated Weibull Mitigation:} We propose an adaptive, randomized anchoring strategy via a Truncated Weibull distribution (customized via shape parameter $k$) to safely explore energy-efficient configurations and dismantle rigid negotiation patterns.
    \item \textbf{non-RT RIC Compatibility via 1B LLMs:} Leveraging the lightweight \texttt{otel-llm-1b-it} model, our multi-agent framework executes complex negotiations with sub-second response times (0.95s), yielding up to 25\% global energy savings while strictly meeting 99.999th percentile URLLC latencies.
\end{itemize}

\section{Network Slicing CVaR Queuing Model}
\label{sec:system_model}

\subsection{System Dynamics and Edge-RAN Queues}

Our architecture (see Figure \ref{arch}) considers a multi-domain network slicing environment encompassing an Edge computing domain and a Radio Access Network (RAN). Service requests for a specific slice $i$ are first processed at the Edge, resulting in a computation latency $L^{\text{edge}}_i$. Subsequently, the processed packets are enqueued for wireless transmission via the RAN, incurring a transmission latency $L^{\text{RAN}}_i$. Consequently, the total end-to-end (E2E) latency is defined as:
\begin{equation}
    L_i = L^{\text{edge}}_i + L^{\text{RAN}}_i.
\end{equation}

Both the Edge and RAN domains operate under finite capacity. Agents representing each slice must negotiate for a partition of the total available RAN bandwidth ($b_{\text{tot}}$) and Edge CPU capacity ($f_{\text{tot}}$). An agent $i$'s decision is represented by the action vector $a_i = (b_i, f_i)$.

\begin{figure}[t]
\centering
\includegraphics[width=0.53\textwidth, trim=0.5cm 0 0 0, clip]{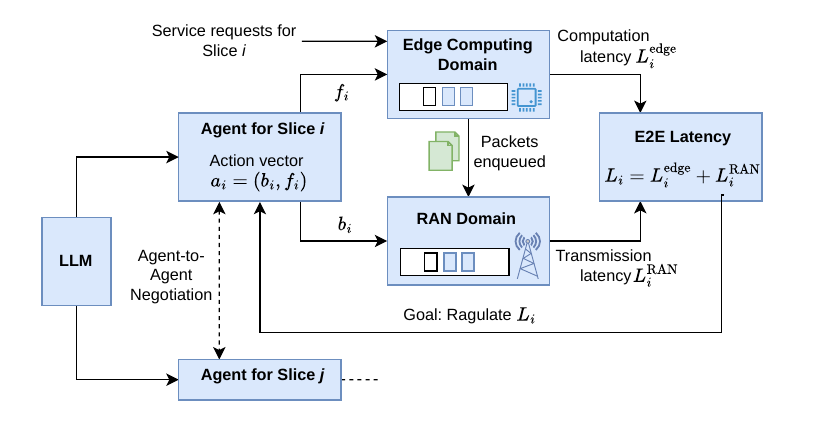}
\caption{Edge-RAN cross-domain slicing model.}
\label{arch}
\end{figure}

To model these dynamics, each agent $i$ maintains a private DT grounded in queuing theory. At each discrete time interval $t$ of duration $\tau$, a volume of bits $\Lambda_{i,t}$ arrives at the Edge according to a time-varying, trial-specific stochastic process:
\begin{equation}
    \Lambda_{i,t} = \lambda_{i,t} \cdot \tau,
\end{equation}
where the mean arrival rate $\mathbb{E}[\lambda_{i,t}]$ is maintained below the service rate to ensure queue stability. The evolution of the Edge computation queue, $Q^{(e)}_{i,t}$, is modeled as:
\begin{equation}
    Q^{(e)}_{i,t+1} = \max\left(0, Q^{(e)}_{i,t} - D^{(e)}_{i,t}\right) + \Lambda_{i,t},
\end{equation}
where $D^{(e)}_{i,t}$ represents the bits processed at the Edge:
\begin{equation}
    D^{(e)}_{i,t} = \tau \cdot C_{i,t}^{(e)}(f_i) = \tau \cdot f_i \cdot C_{\text{CPU}}.
    \label{eq: comp_lat}
\end{equation}
The RAN communication queue, $Q^{(r)}_{i,t}$, is updated based on the output of the preceding computation stage:
\begin{equation}
\begin{split}
    Q^{(r)}_{i,t+1} &= \max\left(0, Q^{(r)}_{i,t} - D^{(r)}_{i,t}\right) \\
    &\quad + \min\left(Q^{(e)}_{i,t} + \Lambda_{i,t}, D^{(e)}_{i,t}\right),
\end{split}
\end{equation}
where $D^{(r)}_{i,t}$ is the volume of transmitted bits, depending on the bandwidth allocation $b_i$ and the stochastic Spectral Efficiency (SE), $\eta_{i,t}$:
\begin{equation}
    D^{(r)}_{i,t} = \tau \cdot C_{i,t}^{(r)}(b_i, \eta_{i,t}) = \tau \cdot b_i \cdot \eta_{i,t}.
    \label{eq: radio_lat}
\end{equation}

Applying Little's Law, we define the average E2E latency $L_{i,T}$ over a horizon $T$ as the ratio of aggregate queue lengths to the average arrival rate:
\begin{equation}
    L_{i,T} = \frac{1}{\mathbb{E}[\Lambda_{i,t}] T} \sum_{t=1}^{T} \left(Q^{(e)}_{i,t} + Q^{(r)}_{i,t}\right).
\end{equation}
The agent's goal is to optimize the action vector $a_i = (b_i, f_i)$ to maintain $L_{i,T}$ within the SLA while minimizing a linear power consumption cost $P_i(a_i)$:
\begin{equation}
    P_i(a_i) = P_{\text{static},i} + C_{\text{BW}} \cdot b_i + C_{\text{CPU}} \cdot f_i,
\end{equation}
where $ C_{\text{BW}}$ and $C_{\text{CPU}}$ are the power consumptions per bandwidth and CPU frequency units, respectively.

\subsection{Digital Twin and CVaR Tail-Latency Prediction}
To guarantee robustness against traffic bursts, the agent's internal DT evaluates proposed actions $a_i$ using the Conditional Value at Risk (CVaR) of the latency distribution, rather than the mean. For an M/M/1 approximation where the sojourn time follows an exponential distribution, the expected shortfall at the $1-\alpha$ confidence level (e.g., $\alpha=0.00001$ for $99.999\%$ URLLC reliability) is formulated as,
\begin{equation}
    \text{CVaR}_{1-\alpha}(L_i) = \mathbb{E}[L_i] \left( 1 - \ln(\alpha) \right).
\end{equation}
Agents negotiate primarily over this strict $\text{CVaR}$ metric. If the negotiated configuration breaches $L_{\text{SLA},i}$, an extreme utility penalty $\mathcal{L}_{\max}$ is incurred.

\section{Agentic Negotiation \& Bimodal Bounds}
\label{sec:agen_negotiation}

Anchoring bias arises when the initial resource proposal of each agent $i$---denoted as the vector $a_i^{(0)} = (b_i^{(0)}, f_i^{(0)})$---systematically influences subsequent allocation updates. This hinders the exploration of feasible multi-resource configurations under the global system constraints in multi-round negotiations, as governed by Algorithm \ref{alg:negotiation}.

\begin{algorithm}[t]
\caption{Adaptive Multi-Agent Resource Negotiation}
\label{alg:negotiation}
\KwIn{Anchor proposals $a_i^{(0)}$}
\KwOut{Agreement $\mathcal{A}$}

Initialize $a_i \leftarrow a_i^{(0)}$\;

\For{$r=1$ \KwTo $R_{\max}$}{

    Evaluate DT-predicted CVaR latency $L_i$ and utility $U_i$\;

    \If{feasible allocation and $U_i\ge U_{\rm th}$ for all agents}{
        Accept agreement\;
    }

    Compute adaptive step size $\delta_r$\;

    Update negotiation context:
    \[
    ctx_i=
    \begin{cases}
    \textsc{DemandInc}, & L_i>L_{{\rm SLA},i},\\
    \textsc{Yield}, & \neg\Phi \vee V_{-i},\\
    \textsc{OptimizeEnergy}, & \text{otherwise}.
    \end{cases}
    \]

    Generate counter-proposal through agentic reasoning and update
    \[
    a_i^{(r+1)}
    =
    a_i^{(r)}
    +
    \delta_r\,\Phi_i(L_i,P_i,\mathrm{DT})
    \]

    subject to capacity constraints.
}

Return agreed or last feasible allocation\;
\end{algorithm}

While anchoring bias is an inherent characteristic of LLM agents reasoning, its operation can be clarified through an analogy with a regularized optimization problem,
\begin{equation}
    a_i^\star = \arg\max_{a_i} \Big(U_i(a_i) - \gamma_i \cdot d(a_i, a_i^{(0)})\Big),
\label{eq:anchor_ran}
\end{equation}
where $d(\cdot,\cdot)$ is a deviation penalty metric and $\gamma_i \geq 0$ quantifies the sensitivity of agent $i$ to its initial proposal. Incorporating the adaptive update logic of Algorithm \ref{alg:negotiation}, the holistic resource adjustment can be assimilated to,
\begin{equation}
    a_i^{(t+1)} = a_i^{(t)} + \delta_t \cdot \Phi_i\big(L_i, P_i, \text{DT}\big) - \gamma_i \big(a_i^{(t)} - a_i^{(0)}\big),
\end{equation}
where $\Phi_i(\cdot)$ encodes the prioritized decision mechanism. 

The impact of anchoring is quantified through utility degradation $\mathcal{L}_{i,\text{anchor}} = U_i(a_i^\dagger) - U_i(a_i^\star)$. Traditional convex bounds assert that $\mathcal{L}_{i,\text{anchor}} \le \frac{\gamma_i^2}{\mu_i} \|a_i^\dagger - a_i^{(0)}\|^2$. However, in highly constrained multi-agent settings with discontinuous SLA penalties, empirical observations contradict these classical bounds. To explain this behavior, we formulate the following theorem:

\begin{theorem}[Bimodal Constraint-Avoidance Utility Bound]
\label{thm:bimodal}
Let $\mathcal{C} = \{ \mathbf{a} \mid \sum_j a_j \le C \}$ be the strict physical capacity constraint of the system, and let the agent's utility function include a penalty cliff $\mathcal{L}_{\max}$ for CVaR SLA violations. The anchoring-induced utility degradation $\mathcal{L}_{i,\text{anchor}}$ relative to the squared anchor distance $d_i = \|a_i^\dagger - a_i^{(0)}\|^2$ exhibits a phase transition characterized by a dual-regime envelope:
\begin{enumerate}
    \item \textbf{Phase 1 (Feasible Convergence):} If $\sum a_j^{(0)} \le C$, the expected utility loss is bounded by a classical convex linear regime:
    \begin{equation}
        \mathcal{L}_i \le \gamma_i d_i.
    \end{equation}
    \item \textbf{Phase 2 (Penalty Recovery):} If $\sum a_j^{(0)} > C$, forcing mandatory algorithmic concessions and triggering the penalty cliff, the expected utility loss decays via an inverse rational envelope as strategic low-balling creates safety buffers:
    \begin{equation}
        \mathcal{L}_i \le \frac{\mathcal{L}_{\max}}{1 + \kappa_i d_i}.
    \end{equation}
\end{enumerate}
\end{theorem}

\begin{proof}
\textbf{Phase 1:} When the initial multi-agent anchor resides within the interior of the feasible set $\text{int}(\mathcal{C})$, the penalty condition is inactive. Assuming the objective utility $U$ is strongly concave with parameter $\mu$, classical optimization distance bounding guarantees that $U(a^\dagger) - U(a^\star) \le \frac{1}{2\mu} \|\nabla U\|^2 \le \gamma_i \|a^\dagger - a^{(0)}\|^2$.

\textbf{Phase 2:} When the initial joint anchor violates $\mathcal{C}$, the negotiation protocol imposes rigid retraction mappings (mandatory concessions) causing severe SLA breaches ($\mathcal{L}_{\max}$). In this regime, an agent can strategically introduce an initial deficit by choosing $a_i^{(0)} \ll a_i^\dagger$, mapping to a spatial safety buffer $s_i \propto \sqrt{d_i}$. Assuming the opponent's stochastic counter-proposals introduce exploration noise with finite variance $\sigma^2$, the probability of the final joint state violating the capacity constraint after the buffer $s_i$ is applied is bounded by Chebyshev's inequality:
\begin{equation}
    \mathbb{P}(\text{Violation}) \le \frac{\sigma^2}{\sigma^2 + s_i^2} = \frac{\sigma^2}{\sigma^2 + c \cdot d_i}.
\end{equation}
Scaling this probability into the utility loss space yields the expected degradation $\mathcal{L}_i = \mathbb{P}(\text{Violation}) \times \mathcal{L}_{\max}$. Letting $\kappa_i = c/\sigma^2$, we arrive at the inverse rational decay envelope $\frac{\mathcal{L}_{\max}}{1 + \kappa_i d_i}$. 
\end{proof}

Theorem \ref{thm:bimodal} provides a profound insight: starting closer to the optimum is actually a \emph{hazard} if it triggers systemic constraint violations.

\section{Weibull Randomized Bias Correction}
\label{sec:proposal}

To navigate the bounds of Theorem \ref{thm:bimodal} and eliminate energy inefficiencies caused by deterministic anchoring, we apply a \emph{Truncated 3-Parameter Weibull Distribution} $\mathcal{W}_{\text{trunc}}(\alpha_c, \beta_c, \gamma_c)$. By diversifying initial proposals across negotiation instances, the system reduces the risk of consistently converging to suboptimal equilibria induced by poor anchors in either bandwidth or computational resources.

We define the spatial constraints for $\mathcal{W}_{\text{trunc}}(\alpha_c, \beta_c, \gamma_c)$ for each resource dimension $c \in \{b, f\}$ as follows: 
\begin{enumerate}
    \item[i)] Lower Bound ($\alpha_c$): The absolute physical floor of the proposal, e.g.,
    \begin{equation}
        \alpha_c = \max(1.0, \; 0.6 \cdot c_{\text{req}}),
    \end{equation}
    \item[ii)] Truncation / Upper Bound ($\beta_c$): The exploratory ceiling, capped to prevent hoarding, e.g.,
    \begin{equation}
        \beta_c = \min(0.9 \cdot c_{\text{tot}}, \; 1.1 \cdot c_{\text{req}}),
    \end{equation}
    \item[iii)] Target Mode ($\gamma_c$): The point of maximum probability density. We intentionally position this mode below the required baseline, e.g.,
    \begin{equation}
        \gamma_c = 0.90 \cdot c_{\text{req}}.
    \end{equation}
\end{enumerate}

\begin{algorithm}[t]
\caption{Initial Resource Proposal with Anchor Bias Mitigation}
\label{alg:anchor_bias}
\KwIn{Slice $i \in \{e, u\}$, Strategy $\mathcal{S}$, SLA limit $L_{\text{SLA},i}$, DT $D_i$, limits $c_{\text{tot}} \in \{b_{\text{tot}}, f_{\text{tot}}\}$}
\KwOut{Initial multi-resource anchor vector $a_i^{(0)} = (b_i^{(0)}, f_i^{(0)})$}
\BlankLine
\For{each resource dimension $c \in \{b, f\}$}{
    \tcp{Find strictly optimal resource for CVaR SLA via DT}
    $c_{\text{req}} \gets \min \big\{ c' \in [c_{\text{abs\_min}}, c_{\text{tot}}] \mid \text{CVaR}_{1-\alpha}(D_i, c') \leq L_{\text{SLA},i} \big\}$\;
    $c_{\text{opt}} \gets c_{\text{req}} / 1.05$\; \tcp*{Target underlying optimum without buffer}
    \uIf{$\mathcal{S} = \text{randomized}$}{
        \tcp{Slice-dependent Truncated Weibull anchor generation}
        \uIf{$i = \text{URLLC}$}{
            $(\alpha_c, \gamma_c, k_i) \gets (0.85 \cdot c_{\text{opt}}, \; 0.98 \cdot c_{\text{opt}}, \; 5.0)$\;
        }
        \Else{
            $(\alpha_c, \gamma_c, k_i) \gets (0.60 \cdot c_{\text{opt}}, \; 0.90 \cdot c_{\text{opt}}, \; 2.0)$\;
        }
        $\beta_c \gets \min(0.9 \cdot c_{\text{tot}}, \; 1.1 \cdot c_{\text{opt}})$\;
        $\lambda_i \gets (\gamma_c - \alpha_c) \cdot \big(\frac{k_i}{k_i-1}\big)^{1/k_i}$\;
        $c_i^{(0)} \sim \mathcal{W}_{\text{trunc}}(\alpha_c, \beta_c, \gamma_c; k_i, \lambda_i)$\;
        $\text{QueryLLM}(\text{Prompt}_{\text{force}}(c_i^{(0)}))$\;
    }
    \Else{
        \tcp{Fixed strategy: default to greedy heuristic}
        $\mathcal{H} \gets \text{Memory.DistillStrategy()}$\;
        $c_i^{(0)} \gets \text{QueryLLM}(\text{Prompt}_{\text{heur}}(\mathcal{H}, c_{\text{target}} \approx c_{\text{req}}))$\;
    }
}
\Return $a_i^{(0)} \gets (b_i^{(0)}, f_i^{(0)})$\;
\end{algorithm}

\begin{figure*}[htbp]
    \centering
    \subfloat[Bimodal Constraint-Avoidance Utility Bounds\label{fig:util_bounds}]{%
      \includegraphics[
        width=0.35\textwidth,
        trim=0.8cm 0.4cm 0.5cm 0.4cm,
        clip
      ]{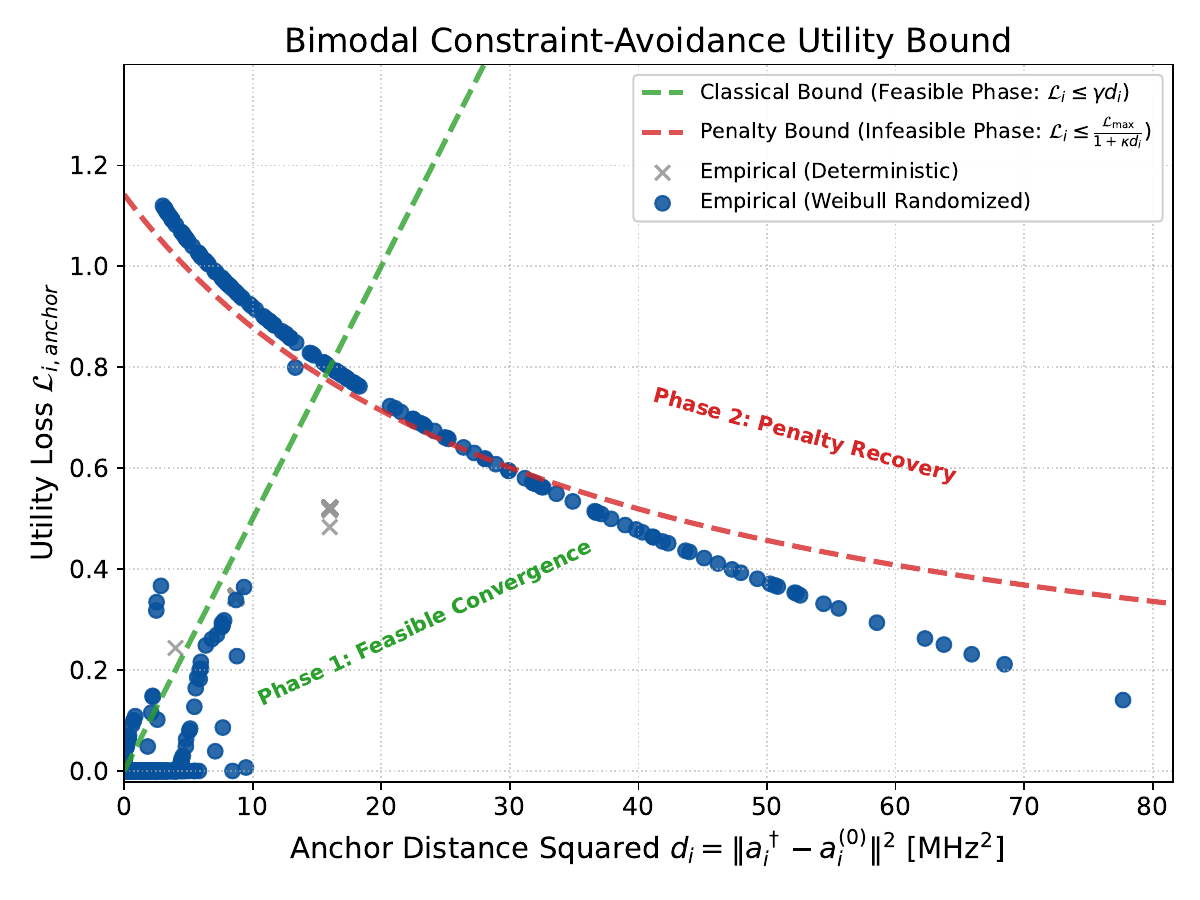}
    }
    \hspace{0.03\textwidth}
    \subfloat[Utility Degradation CDF\label{fig:util_cdf}]{%
      \includegraphics[
        width=0.35\textwidth,
        trim=0.8cm 0.4cm 0.5cm 0.4cm,
        clip
      ]{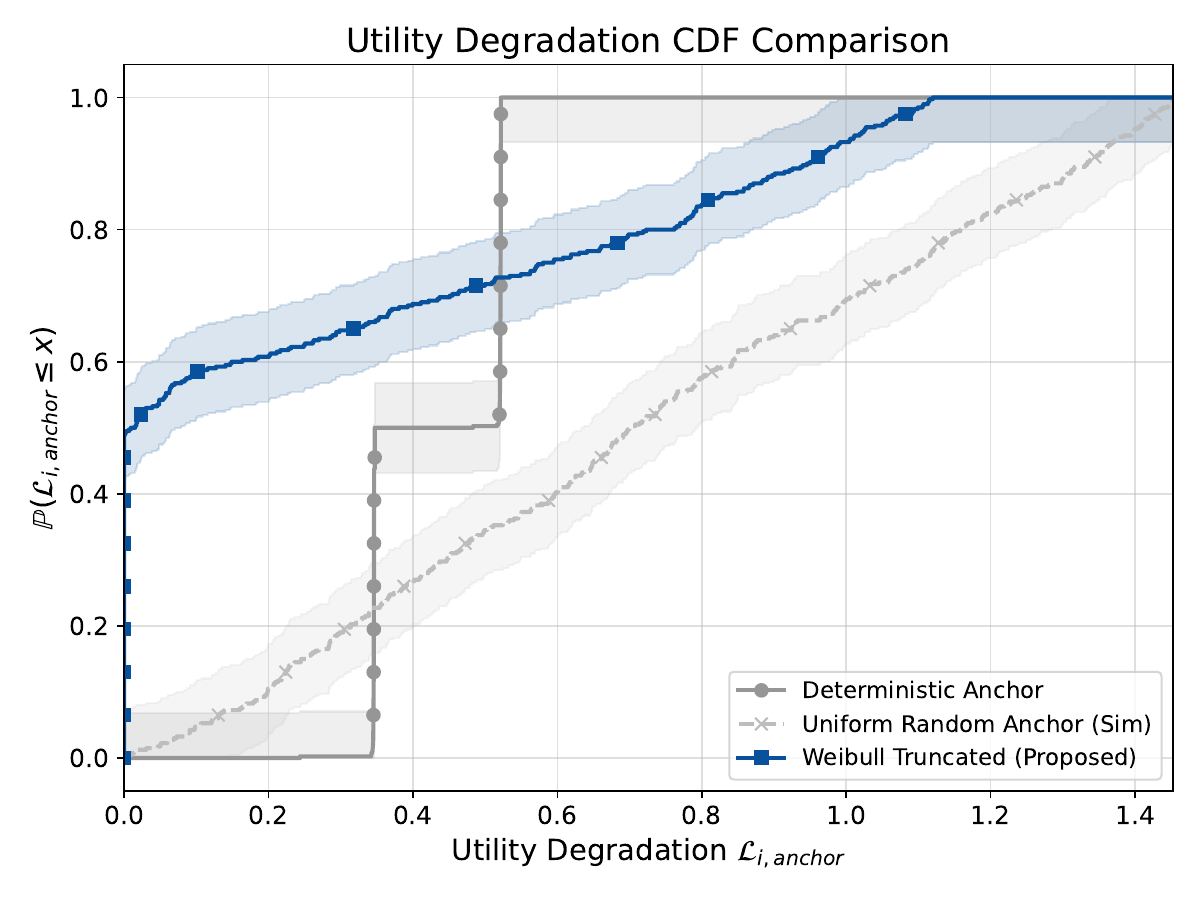}
    }
    \caption{Analytical validation of anchoring-induced utility degradation ($\mathcal{L}_{i,\text{anchor}}$) mapped against the squared anchor distance ($d_i$). The empirical data distinctly traces the \emph{Bimodal Constraint-Avoidance} bounds proven in Theorem \ref{thm:bimodal}.}
    \label{fig:analytical_results}
\end{figure*}
The initial resource proposal $c_i^{(0)}$ is drawn from the continuous random variable $X \sim \mathcal{W}_{\text{trunc}}(\alpha_c, \beta_c, \gamma_c)$. To ensure the peak of the distribution aligns exactly with our target mode $\gamma_c$, we derive the scale parameter $\lambda$ as:
\begin{equation}
    \lambda = \left(\gamma_c - \alpha_c\right) \times \left(\frac{k}{k-1}\right)^{\frac{1}{k}}.
\end{equation}

The resulting Truncated Weibull PDF $f(x)$ is formally defined over the support $\alpha_c \le x \le \beta_c$ as:
\begin{equation}
\begin{aligned}
f(x) &= \frac{1}{1 - e^{-\left(\frac{\beta_c-\alpha_c}{\lambda}\right)^k}} \\
     &\times \frac{k}{\lambda} \left(\frac{x-\alpha_c}{\lambda}\right)^{k-1} \exp\left(-\left(\frac{x-\alpha_c}{\lambda}\right)^k\right),
\end{aligned}
\end{equation}
and 0 otherwise. The deployment of this strategy is shown in Algorithm \ref{alg:anchor_bias}.

We utilize adaptive shape parameters ($k$) to reflect specific slice tolerances. For URLLC, which possesses extreme sensitivity to tail-latencies, we utilize a tightly peaked distribution ($k=5.0$) with a very safe lower limit ($\alpha_c = 0.85 \cdot c_{\text{req}}$). For eMBB, we encourage wider exploration ($k=2.0$, $\alpha_c = 0.6 \cdot c_{\text{req}}$). The Weibull PDF exponentially penalizes overly aggressive demands via its $e^{-x^k}$ right-tail decay, ensuring initial proposals inject bounded variance into the negotiation while rigorously anchoring the system toward the optimal penalty-recovery envelope (Phase 2 of Theorem \ref{thm:bimodal}).

\section{Numerical Results}
\label{sec:exp_setup}

\subsection{Experimental Setup}
We evaluate the mitigation of anchoring bias by dynamically allocating shared RAN bandwidth ($b_{\text{tot}} = 60$ MHz) and Edge CPU capacity ($f_{\text{tot}} = 40$ GHz) between two network slices. The simulation advances in discrete time steps of $\tau = 10\text{ ms}$. To emulate realistic 6G network dynamics, we deploy differentiated, time-varying traffic profiles: the eMBB slice generates a base load of 90~Mbps subject to slow sinusoidal fluctuations, whereas the URLLC slice generates a base load of 40~Mbps interspersed with deterministic 50\% traffic bursts to stress the system's tail-latency handling. The autonomous agents are powered by the lightweight \texttt{otel-llm-1b-it} model hosted locally \cite{otel2026}, interfacing via an OpenAI-compatible API with a temperature setting of $0.1$ to ensure highly focused, JSON-compliant reasoning. The negotiation framework restricts interaction to a maximum of $R_{\max}=5$ rounds per trial. Slice latency targets are evaluated at the rigorous CVaR tail: $50\text{ ms}$ for eMBB ($95\%$ confidence) and $10\text{ ms}$ for URLLC ($99.999\%$ confidence). We contrast a deterministic \emph{Fixed} baseline strategy---which anchors heavily on conservative heuristics---against the proposed \emph{Weibull Randomized} strategy over 200 independent trials. Each trial concludes with a 10-step post-agreement evaluation phase to rigorously measure SLA violations, queue dynamics, and aggregate energy savings against the maximum baseline power consumption.

\begin{figure}[t]
    \centering
    \includegraphics[width=0.48\textwidth]{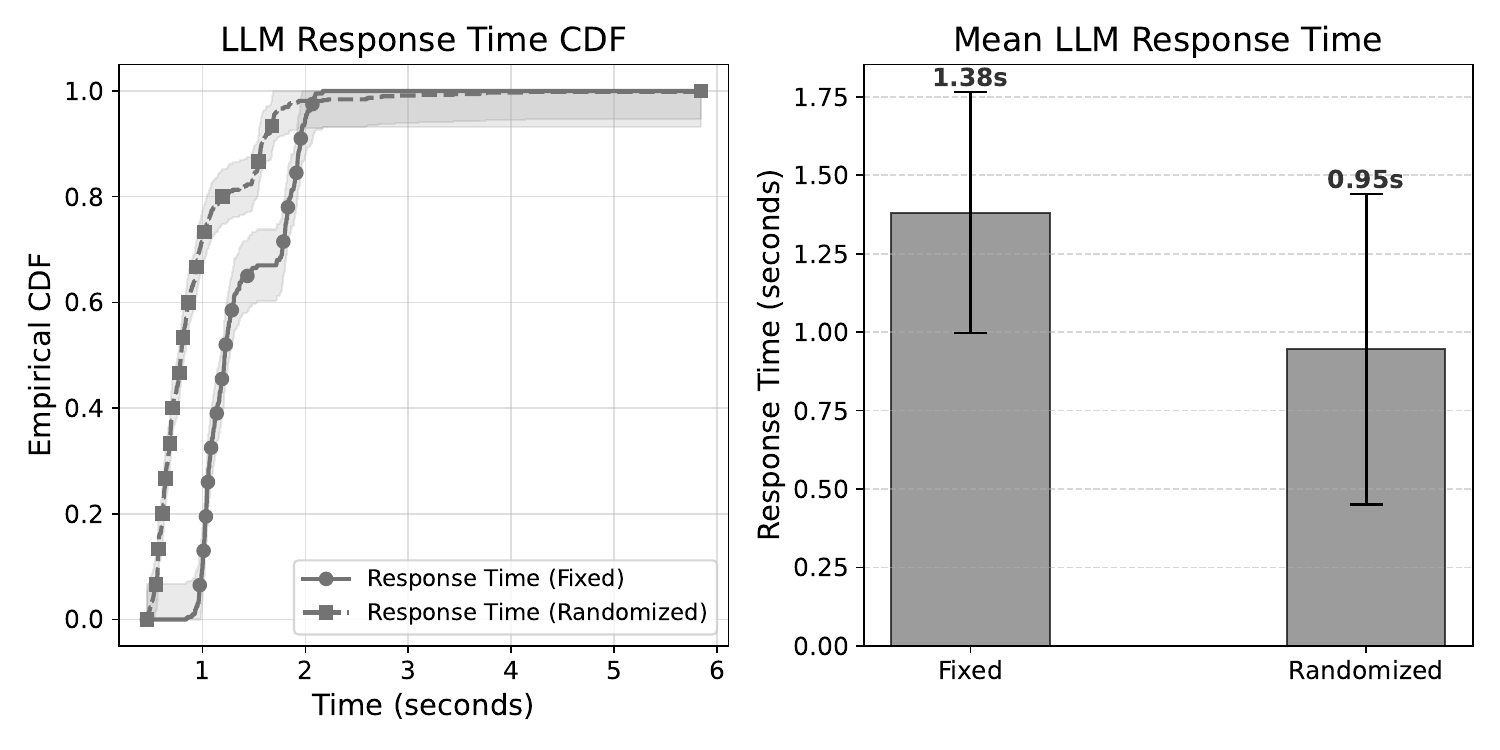}
    \caption{The use of a locally hosted 1B model ensures sub-second responses compatible with the non-RT RIC.}
    \label{fig:llm_time}
\end{figure}

\subsection{Analytical Validation of Theorem 3.1}
The mechanics of our proposed Bimodal Constraint-Avoidance Utility Theorem are empirically validated in Figure \ref{fig:analytical_results}. Figure \ref{fig:util_bounds} plots the empirical utility degradation ($\mathcal{L}_{i,\text{anchor}}$) against the squared anchor distance ($d_i = \|a_i^\dagger - a_i^{(0)}\|^2$). As predicted by Theorem \ref{thm:bimodal}, the empirical data distinctly bifurcates into two mathematical regimes, demonstrating a clear phase transition based on initial feasibility.
\begin{figure}[t]
    \centering
    \subfloat[Energy Saving CDF\label{fig:eng_cdf}]{%
      \includegraphics[width=0.35\textwidth]{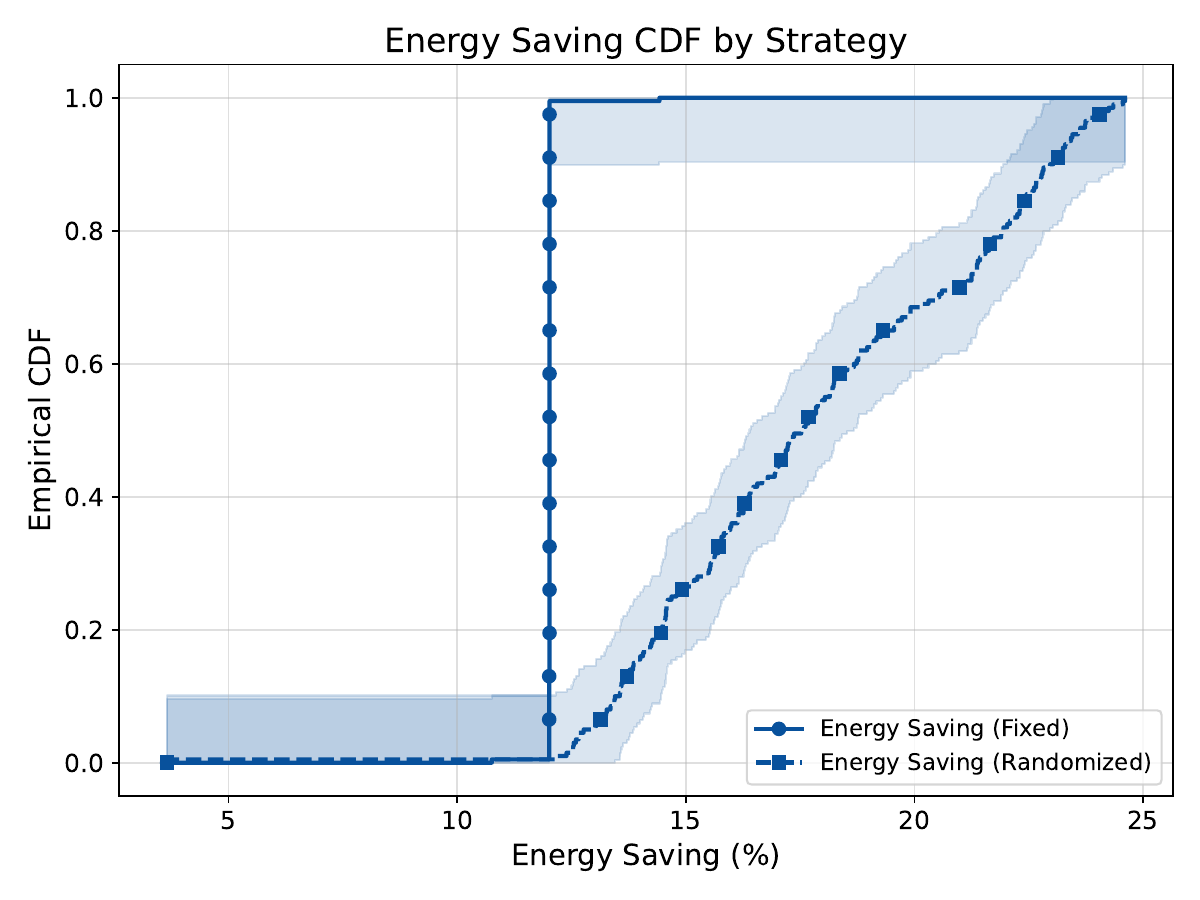}
    }
    \vfill
    \subfloat[Latency CDF by Strategy\label{fig:lat_cdf}]{%
      \includegraphics[width=0.35\textwidth]{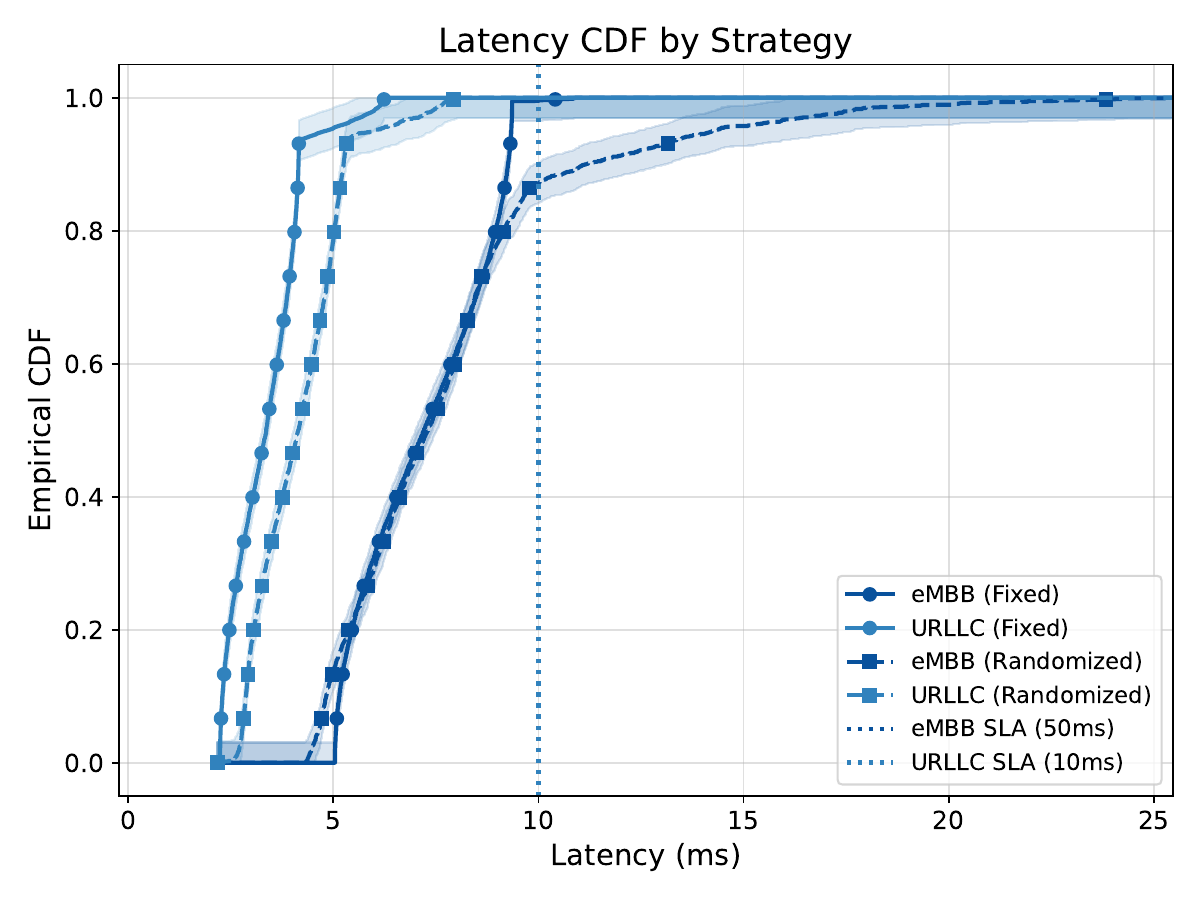}
    }
    \vfill
    \subfloat[Latency Statistics (Mean \& 99.999th Percentile)\label{fig:lat_stats}]{%
      \includegraphics[width=0.35\textwidth]{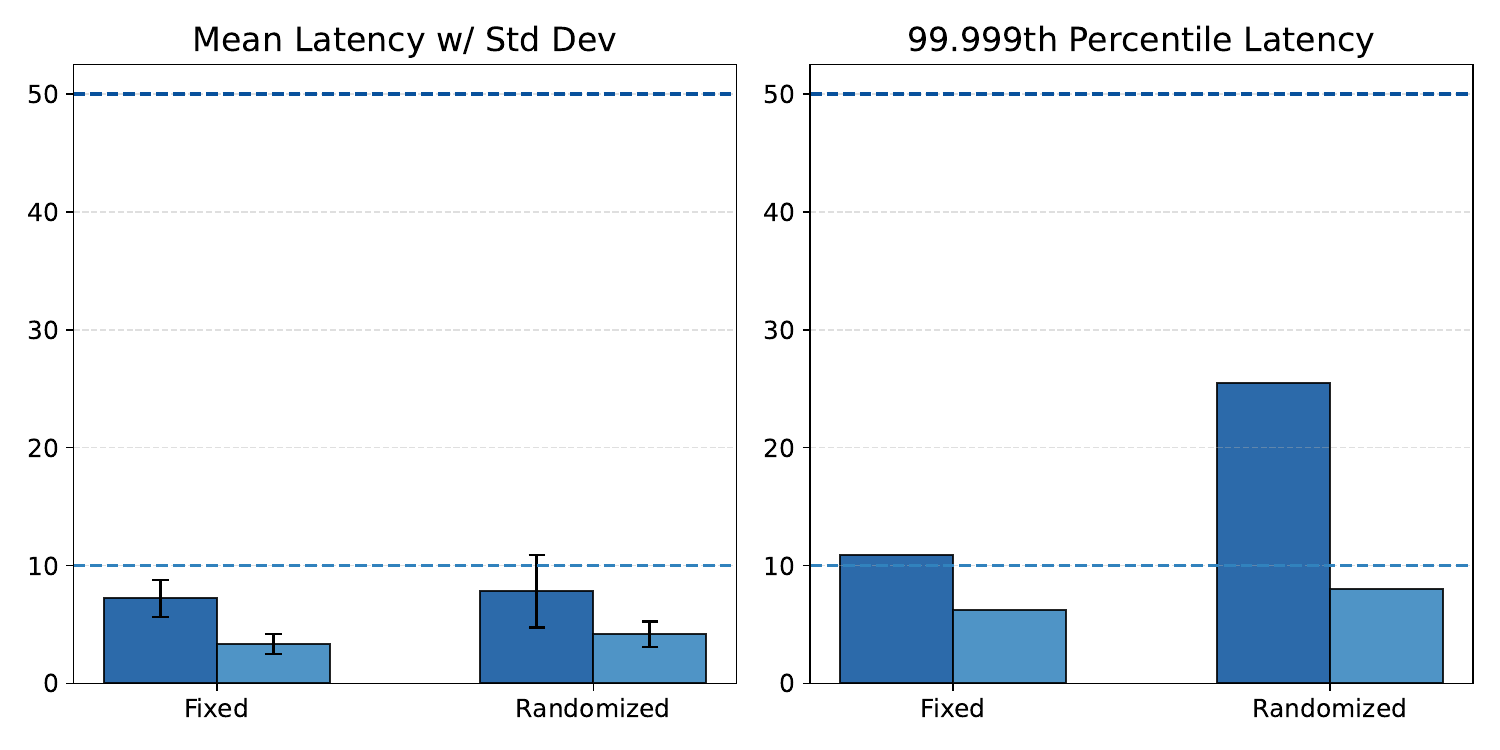}
    }
    \caption{Overall system performance comparing Deterministic and Weibull Randomized anchoring.}
    \label{fig:system_performance}
\end{figure}
First, for instances where the joint initial anchor lies within the physical RAN capacity limit ($\sum a_j^{(0)} \le 60$ MHz), the negotiation proceeds without triggering severe SLA penalty cliffs. As shown by the dense cluster in the bottom-left corner of Figure \ref{fig:util_bounds}, these points trace the \emph{Phase 1 Feasible Convergence} regime. Here, the utility loss is strictly governed by the suboptimality of the initial distance, tightly hugging the linearly rising Classical Bound defined by $\mathcal{L}_i \le \gamma_i d_i$ (with an empirical slope of $\gamma_i \approx 0.05$).

Conversely, when agents propose aggressively greedy initial anchors that exceed the 60 MHz limit, the system enters the \emph{Phase 2 Penalty Recovery} regime. The rigid capacity constraint forces immediate, mandatory algorithmic concessions, introducing a high risk of violating the strict CVaR SLA targets and triggering a maximum utility penalty ($\mathcal{L}_{\max} \approx 1.2$). However, as mathematically formalized in Phase 2 of the theorem, when an agent strategically injects an initial deficit (increasing $d_i$ via tactical "low-balling"), it creates a spatial safety buffer that reduces the probability of a forced constraint violation. The empirical data flawlessly traces this phenomenon: as the anchor distance increases, the utility degradation smoothly slides down the Inverse Rational decay envelope ($\mathcal{L}_i \le \frac{\mathcal{L}_{\max}}{1 + \kappa_i d_i}$, governed by $\kappa_i = 0.03$). 

The proposed Truncated Weibull strategy demonstrates superior adaptability across both domains. While the deterministic baseline (grey crosses) frequently traps agents in suboptimal regions or high-penalty states, the Weibull-driven agents (blue circles) successfully ride the decay curve to mitigate losses. Consequently, Figure \ref{fig:util_cdf} demonstrates the macroscopic impact of this behaviour: the Weibull approach substantially left-shifts the degradation Cumulative Distribution Function (CDF). By effectively mapping its anchor generation to the mathematically optimal regions of the bimodal envelope, the proposed strategy minimizes median utility loss and completely eradicates the extreme tail risks inherent to uniform or deterministic exploration.

\subsection{System Performance \& non-RT RIC Compatibility}
For practical O-RAN integration, Figure \ref{fig:llm_time} showcases the negotiation latency. Because we utilize the optimized 1B parameter model, the mean LLM response time drops to 0.95 seconds under the randomized strategy. This sub-second latency footprint comfortably places complex agentic negotiations within the 1-second timescale required for non-RT RIC operations. As shown in Figure \ref{fig:eng_cdf}, the Weibull strategy breaks the rigid energy savings limit of the deterministic model, dynamically shedding excess capacity to push system-wide energy savings up to a maximum of 25\%, establishing a superior, sustainable equilibrium. Furthermore, a deliberate strategic trade-off is observable. Figure \ref{fig:lat_stats} confirms that both strategies meet strict SLA limits, notably successfully satisfying the 99.999th percentile CVaR latency for the URLLC slice below 10ms. However, the randomized strategy intentionally utilizes the available SLA slack (evidenced by higher mean latencies in Figure \ref{fig:lat_cdf}). This consumption of slack is converted directly into massive energy efficiency gains. Ultimately, the integration of bias-mitigated agentic reasoning within the non-RT RIC heralds a shift toward intent-driven, zero-touch network orchestration. By replacing rigid heuristic baselines with the dynamically adaptive Truncated Weibull anchor, the 6G system not only achieves deterministic tail-latency guarantees under stochastic loads but also maximizes resource frugality. The ability to perform such high-level, multi-objective optimization using a locally hosted, privacy-preserving 1B model proves that cognitive de-biasing is both computationally viable and strictly necessary for the sustainable operation of future autonomous networks.

\section{Conclusion}
\label{sec:conclusion} 

This paper addressed anchoring bias in autonomous, LLM-driven 6G network slicing. By integrating CVaR tail-latency evaluations and proposing the Bimodal Constraint-Avoidance Theorem, we rigorously mapped the dual-regime nature of multi-agent utility degradation. Leveraging a Truncated Weibull randomization strategy with adaptive shape scaling, we systematically dismantled rigid over-provisioning heuristics. Empirical deployments using the lightweight \texttt{otel-llm-1b-it} model yielded highly performant sub-second inferences (0.95s mean), ensuring seamless compatibility with non-RT RIC control loops. This mathematically bounded exploration enabled agents to navigate toward SLA boundaries confidently, boosting network-wide energy savings to 25\% while strictly maintaining 99.999th percentile URLLC reliability. Ultimately, this framework accelerates the realization of intent-driven, zero-touch O-RAN orchestration. We eventually proved that privacy-preserving, locally hosted small LLMs can execute such complex optimizations under cognitive de-biasing, forming a fundamental pillar for sustainable 6G autonomous networks.

\bibliographystyle{IEEEtran}
\bibliography{bibliography}

\end{document}